\documentclass[11pt]{article}

\pdfoutput=1 %

\usepackage{amsmath,amsthm,amssymb,mathtools} %
\usepackage{authblk} %
\usepackage[english]{babel} %
\usepackage{booktabs} %
\usepackage{color} %
\usepackage{dsfont} %
\usepackage{enumitem} %
\usepackage{flushend} %
\usepackage[T1]{fontenc} %
\usepackage{framed} %
\usepackage{graphicx} %
\usepackage[hmargin=1in,vmargin=1in]{geometry} %
\usepackage{hyphenat} %
\usepackage{mathtools} %
\usepackage{microtype} %
\usepackage[utf8]{inputenc} %
\usepackage{soul} %
\usepackage{subdepth} %
\usepackage{suffix} 
\usepackage{xspace} %
\usepackage{csquotes} %
\usepackage{hyperref} %
\usepackage{wasysym} 

\setul{1ex}{.5pt} %
\definecolor{shadecolor}{gray}{.925} %

\newtheorem{theorem}{Theorem} %
\newtheorem{lemma}[theorem]{Lemma} %
\newtheorem{definition}{Definition} %
\newtheorem{fact}{Fact} %

\newcommand{\microspace}{\mspace{.5mu}} %
\newcommand{\ket}[1]{\ensuremath{\lvert\microspace #1
    \microspace\rangle}} %
\newcommand{\bigket}[1]{\bigl\lvert\microspace #1
  \microspace\bigr\rangle} %
\newcommand{\Bigket}[1]{\Bigl\lvert\microspace #1
  \microspace\Bigr\rangle} %
\newcommand{\bra}[1]{\ensuremath{\langle\microspace #1
    \microspace\rvert}} %
\newcommand{\bigbra}[1]{\bigl\langle\microspace #1
  \microspace\bigr\rvert} %
\newcommand{\ip}[2]{\ensuremath{\left\langle#1,#2\right\rangle}} %
\newcommand{\braket}[2]{\ensuremath{\left\langle#1 \middle\vert
      #2\right\rangle}} %
\newcommand{\norm}[1]{\ensuremath{\left\lVert #1 \right\rVert}} %
\newcommand{\abs}[1]{\ensuremath{\left\lvert #1 \right\rvert}} %
\newcommand{\defeq}{\stackrel{\smash{\text{\tiny\rm def}}}{=}} %

\newcommand*{\dif}{\mathop{}\!\mathrm{d}} %
\newcommand{\complex}{\mathbb{C}} %

\newcommand{\class}[1]{\textup{#1}\xspace} %

\newcommand{\QMIP}{\class{QMIP}} %
\WithSuffix\newcommand\QMIP*{\ensuremath{\class{QMIP}^*}} %
\newcommand{\MIP}{\class{MIP}} %
\WithSuffix\newcommand\MIP*{\ensuremath{\class{MIP}^*}} %

\newcommand{\reg}[1]{\textsf{#1}} %
\newcommand{\dtr}{\operatorname{D}} 

\newcommand{\setft}[1]{\mathrm{#1}} %
\newcommand{\Density}{\setft{D}} %
\newcommand{\Unitary}{\setft{U}} %
\newcommand{\Lin}{\setft{L}} %
\newcommand{\State}{\setft{S}} %

\newcommand{\secpar}{\kappa} %
\newcommand{\secparam}{1^\secpar} %
\newcommand{\alg}[1]{\textsf{#1}\xspace} %

\newcommand{\KeyGen}{\alg{KeyGen}} %
\newcommand{\Bank}{\alg{Bank}} %
\newcommand{\Ver}{\alg{Ver}} %
\newcommand{\Count}{\alg{Count}} %

\newcommand{\money}{\$} %

\newcommand{\prf}{\alg{PRF}} %
\newcommand{\prp}{\alg{PRP}} %
\newcommand{\prs}{\alg{PRS}} %
\newcommand\Sym{\textrm{Sym}} %

\DeclareMathOperator{\poly}{poly} %
\DeclareMathOperator{\negl}{negl} %

\DeclareMathOperator{\tr}{tr} %
\def\I{\mathds{1}} 

\def\A{\mathcal{A}} %
\def\B{\mathcal{B}} %
\def\C{\mathcal{C}} %
\def\D{\mathcal{D}} %
\def\H{\mathcal{H}} %
\def\K{\mathcal{K}} %
\def\S{\mathcal{S}} %
\def\X{\mathcal{X}} %
\def\Y{\mathcal{Y}} %

\def\x{\mathbf{x}} %
\def\y{\mathbf{y}} %

\newcommand{\E}{\mathop{\mathbb{E}}\displaylimits} 

\begin{document}


\title{\LARGE\bf Pseudorandom States, Non-Cloning Theorems and Quantum
  Money}

\author[1]{Zhengfeng Ji} %
\author[2]{Yi-Kai Liu} %
\author[3]{Fang Song} %

\affil[1]{Centre for Quantum Software and Information,
  Faculty of Engineering and Information Technology\protect\\
  University of Technology Sydney, Ultimo 2007, NSW,
  Australia\vspace{2mm}}

\affil[2]{Joint Center for Quantum Information and Computer Science (QuICS),\protect\\
  National Institute of Standards and Technology (NIST), Gaithersburg, MD, USA,\protect\\
  and University of Maryland, College Park, MD, USA\vspace{2mm}}

\affil[3]{Computer Science Department, Portland State University,
  Portland 97201, OR, USA}

\date{}
\renewcommand\Affilfont{\normalsize\itshape} %
\renewcommand\Authfont{\large} %
\renewcommand\Authsep{\rule{8mm}{0mm}} %
\renewcommand\Authands{\rule{8mm}{0mm}} %

\maketitle

\thispagestyle{empty}

\begin{abstract}
  We propose the concept of pseudorandom states and study their
  constructions, properties, and applications.
  Under the assumption that quantum-secure one-way functions exist, we
  present concrete and efficient constructions of pseudorandom states.
  The non-cloning theorem plays a central role in our study---it
  motivates the proper definition and characterizes one of the
  important properties of pseudorandom quantum states.
  Namely, there is no efficient quantum algorithm that can create more
  copies of the state from a given number of pseudorandom states.
  As the main application, we prove that any family of pseudorandom
  states naturally gives rise to a private-key quantum money scheme.
\end{abstract}

\section{Introduction}
\label{sec:intro}


Pseudorandomness is one of the foundational concepts in modern
cryptography and theoretical computer science.
A distribution function, or a permutation is pseudorandom if it is
\emph{computationally indistinguishable} from a corresponding truly
random object~\cite{Sha83,Yao82,BM84}.
Pseudorandom objects, such as pseudorandom generators (PRGs),
pseudorandom functions (PRFs) and pseudorandom permutations (PRPs) are
fundamental and ubiquitous cryptographic building blocks in the design
of stream ciphers, block ciphers and message authentication
codes~\cite{GGM86,LR88,GGM85,Rom90,HILL99}.  In complexity theory,
pseudorandomness is of vital importance to an area called
derandomization~\cite{NW94,IW97}.


In quantum information, truly random bits can be generated easily with
trusted or even untrusted quantum devices.
Is pseudorandomness, a seemingly weaker notion of randomness, still
relevant in the context of quantum information processing?
The answer is yes.
Pseudorandom objects are usually much more computationally efficient.
By simple counting argument, one needs exponentially many bits
even to specify a truly random function or permutation on $n$-bit
strings.
Hence, truly random objects are not feasible in most of the
cryptography applications.
In this sense, we should think of pseudorandomness not as a weaker but
as a different variant of randomness with its own characteristics and
strength.


There are recent studies of pseudorandom objects from quantum
information perspectives motivated by their applications in
post-quantum cryptography.
One natural question is whether the classical constructions such as
PRFs and PRPs remain secure in the quantum setting.
This is a challenging task as, for example, a quantum adversary may
query the underlying function or permutation in \emph{superposition}.
Fortunately, people have so far restored a lot of positive results.
Assuming a one-way function that is hard to invert for polynomial-time
quantum algorithms, we can attain quantum-secure PRGs as well as
PRFs~\cite{HILL99,Zha12}.
Furthermore, assuming quantum-secure PRFs, one can construct quantum
secure PRPs using various \emph{shuffling}
constructions~\cite{Zha16,Song17}.


A related area of development in quantum information concerns
\textit{state $t$-designs} and \textit{unitary
  $t$-designs}~\cite{AE07,DCEL09,HL09,BHH16,CLLW16,KG15,Web16,Zhu15}.
These are quantum analogues of $t$-wise independent random variables,
which are another useful relaxation of randomness.
The major difference between $t$-wise independence and
pseudorandomness is the following.
In the case of $t$-wise independence, the observer who receives the
randomness is only given a fixed number of samples, but may be
computationally unbounded; hence quantum $t$-designs satisfy an
``information theoretic'' or ``statistical'' notion of security.
In contrast, in the case of pseudorandomness, the observer who
receives the randomness is assumed to be computationally efficient;
this leads to a ``computational'' notion of security, based on some
complexity theoretic assumption (e.g., the existence of one-way
functions).

In general, these two notions, $t$-wise independence and
pseudorandomness, are incomparable.
On one hand, the setting of pseudorandomness imposes stronger
restrictions on the observer, since it assumes a bound on the
observer's computational effort (say, running in probabilistic
polynomial time).
On the other hand, the setting of $t$-wise independence imposes
stronger restrictions on the observer, since it forces the observer to
be non-adaptive and limits the number of input copies to the fixed
parameter $t$, which is usually a constant or a fixed polynomial.
In addition, different distance measures are often used, e.g., trace
distance or diamond norm, versus computational distinguishability.


In this work, we study pseudorandom \emph{quantum} objects such as
quantum states and unitary
operators.
Quantum states (in analogy to
strings) and unitary operations (in analogy to functions) form
continuous spaces, and a ``random'' state or unitary operation
inevitably exhibits unique as well as puzzling features.  Formally,
people consider the Haar measure to capture perfect randomness on the
spaces of quantum states and unitary operators.  A natural basic
question is:
\begin{center}
  \itshape How to define and construct \emph{computational}
  \emph{pseudorandom} approximations of Haar randomness,\\ and what
  are their applications?
\end{center}


\noindent
\textit{Our contributions.}
We set forth to attack this question formally.
To this end, we propose definitions of pseudorandom quantum states
(PRS's) and pseudorandom unitary operators (PRUs), present efficient
constructions of PRS's, and demonstrate their applications such as
private-key quantum money schemes and quantum thermalization.
Our main contributions include:
\begin{enumerate}
\item We propose a suitable definition of \emph{quantum pseudorandom
    states}.

  What is a proper definition of PRS's?
  It is obvious that we should employ the notion of quantum
  \emph{computational indistinguishability} in the definition of
  quantum pseudorandom states.
  Roughly, we consider a collection of quantum states $\bigl\{ \ket{\phi_k}
  \bigr\}$ indexed by $k\in\K$; and as in the definition of
  pseudorandom distributions, we may require that no efficient quantum
  algorithm can tell the difference between $\ket{\phi_k}$ for a random
  $k$ and a state drawn according to the Haar measure.
  This is a reasonable definition that has been seriously considered in
  the literature (e.g.~\cite{BMW09,CCL+17}).
  Although this definition may be applicable in certain situations, it
  does not seem to grasp the quantum nature of the problem as purely
  classical distributions\footnote{For example, a uniform distribution
    over the computational basis state $\{\ket{k}\}$ has identical
    density matrix as the Haar random states.
    Yet, there is however absolutely no quantum phenomena in this
    family of states and most of the interesting applications
    discussed later in this paper become impossible.}
  may also satisfy the definition.
  Instead, we require that any adversary cannot tell the difference
  even given any \emph{polynomially many} copies of the state.
  This stronger definition, motivated by the non-cloning theorem, will
  in turn give rise to an interesting non-cloning theorem for
  pseudorandom states.
  One can argue that this definition is also a faithful generalization
  of pseudorandom distributions to the quantum state setting---the
  access to many copies is not explicitly mentioned classically as one
  can always make arbitrarily many copies of classical information.
  It is also consistent with the definition of pseudorandom Unitary
  operators discussed later in this paper.

\item We present concrete efficient constructions of PRS's with the
  minimal assumption that quantum-secure one-way functions exist.

  Our construction uses any quantum-secure $\prf=\{\prf_k\}_{k\in\K}$
  and computes it into the phases of a uniform superposition state.
  We call such family of PRS the \textit{random phase states}.
  This family of states can be efficiently generated using the quantum
  Fourier transform and a phase kick-back trick.
  We prove that this family of state is pseudorandom by a hybrid
  argument.
  By the quantum security of \prf, the family is computationally
  indistinguishable from a similar state family defined by truly
  random functions.
  We then prove that, this state family corresponding to truly random
  functions is statistically indistinguishable from Haar random
  states.
  Finally, by the fact that PRF exists assuming quantum-secure one-way
  functions, we can base our PRS construction on quantum-secure
  one-way functions.

\item We establish interesting properties of PRS's and discuss several
  applications.
  These include the \emph{cryptographic non-cloning theorems} for
  PRS's and the construction of private-key quantum money schemes
  based on PRS's.

  We prove that a PRS remains pseudorandom, even if we additionally
  give the distinguisher an oracle that reflects about the given state
  (i.e., $O_\phi: = \I - 2\ket{\phi}\bra{\phi}$).
  This establishes the equivalence between the standard and a strong
  definition of PRS's.
  Technically, this is proved using the fact that with polynomially
  many copies of the state, one can approximately simulate the
  reflection oracle $O_\phi$.

  We obtain general \emph{cryptographic non-cloning theorems} of PRS's
  both with and without the reflection oracle.
  The theorems roughly state that given any polynomially many copies of
  pseudorandom states, no polynomial-time quantum algorithm can
  produce even one more copy of the state.
  We call them cryptographic non-cloning theorems due to the
  background of PRS's in pseudorandomness.
  The proofs of these theorems use SWAP tests in the reduction from a
  hypothetical cloning algorithm to an efficient distinguishing
  algorithm violating the definition of PRS's.

  Using the strong pseudorandomness and the cryptographic non-cloning
  theorem with reflection oracle, we are able to give simple proofs
  that the corresponding private-key quantum money scheme is secure.
  We stress that provably secure quantum money schemes had been
  elusive until Aaronson and Christiano finally proved the first
  secure scheme in the black-box setting in 2012~\cite{AC12}.  They
  used a specific construction based on hidden subspace state, whereas
  our construction is more generic and can be based on any PRS.  The
  freedom to choose and tweak the underlying pseudorandom functions or
  permutations in the PRS may motivate and facilitate the construction
  of public-key quantum money schemes in future work.  Our proof also
  takes an arguably simpler route than that in~\cite{AC12}.

  In general, PRS's may be used in place of high-order quantum
  $t$-designs, giving a performance improvement in certain
  applications.
  For example, pseudorandom states can be used to construct toy models
  of quantum \emph{thermalization}, where one is interested in quantum states
  that can be prepared efficiently (via some dynamical process), yet
  have ``generic'' or ``typical'' properties (as exemplified by
  Haar-random pure states, for instance)~\cite{PSW06}.
  Using $t$-designs with polynomially large $t$, one can construct
  states that are ``generic'' in a strong information-theoretic
  sense~\cite{Low09}.
  Using PRS, one can construct states that satisfy a weaker property:
  they are computationally indistinguishable from ``generic'' states,
  for a polynomial-time observer.
  But the PRS states may be more physically plausible, because they
  can be prepared in a shorter time (e.g., by a polylogarithmic-depth
  quantum circuit).

\item We propose a definition of \emph{quantum pseudorandom unitary
    operators} (PRUs).
  We also present candidate constructions of PRUs (without a proof of
  security), by extending our techniques for constructing PRS's.

\end{enumerate}

\begin{table}[htb!]
  \centering
  \begin{tabular}{c c c}
    \toprule
    & \textbf{Classical} & \textbf{Quantum}\\
    \midrule
    Ideal & Uniform & Haar measure\\
    \midrule
    $t$-wise independence & $t$-wise independence & quantum $t$-designs \\
    \midrule
    & & (\textit{this work})\\
    Pseudorandom & PRGs & PRS's \\
    & PRFs, PRPs & PRUs \\
    \bottomrule

  \end{tabular}
  \caption{Summary of various notions that approximate perfect
    randomness}
  \label{tab:summary}
\end{table}


\noindent
\textit {Discussion}.
We summarize the mentioned variants of randomness in
Table~\ref{tab:summary}.
The focus of this work is mostly about PRS's and we briefly touch upon
PRUs.
We view our work as an initial step in an exciting direction and
expect more applications and new questions inspired by our notion of
pseudorandom states and unitary operators.

Here, we point out some open problems which are particularly important
and interesting.
First, can we prove the security of our candidate PRU constructions?
The techniques developed in quantum unitary designs~\cite{HL09,BHH16}
seem helpful.
Second, it is a natural question to ask whether quantum-secure one-way
functions are necessary for the construction of PRS's.
Third, can we establish security proofs for more candidate
constructions of PRS's?
Different constructions may have their own special properties that may
be useful in different settings.
Finally, it is interesting to explore whether our quantum money
construction may be adapted to a public-key money scheme under
reasonable cryptographic assumptions.

\section{Preliminaries}
\label{sec:prelim}


\subsection{Notions}

For a finite set $\X$, $\abs{\X}$ denotes the number of elements in
$\X$.
We use the notion $\Y^\X$ to denote the set of all functions
$f:\X\rightarrow\Y$.
For finite set $\X$, we use $x\gets\X$ to mean that $x$ is drawn
uniformly at random from $\X$.
The permutation group over elements in $\X$ is denoted as $S_\X$.
We use $\poly(\secpar)$ to denote the collection of polynomially
bounded functions of the security parameter $\secpar$, and use
$\negl(\secpar)$ to denote negligible functions in $\secpar$.
A function $\epsilon(\secpar)$ is \emph{negligible} if for all constant
$c>0$, $\epsilon(\secpar) < \secpar^{-c}$ for large enough $\secpar$.

In this paper, we use a \emph{quantum register} to name a collection
of qubits that we view as a single unit.
Register names are represented by capital letters in a \emph{sans
  serif} font.
We use $\State(\H)$, $\Density(\H)$, $\Unitary(\H)$ and $\Lin(\H)$ to
denote the set of pure quantum states, density operators, unitary
operators and bounded linear operators on space $\H$ respectively.
An ensemble of states $\{(p_i, \rho_i)\}$ represents a system prepared
in $\rho_i$ with probability $p_i$.
If the distribution is uniform, we write the ensemble as $\{\rho_i\}$.
The adjoint of matrix $M$ is denoted as $M^*$.
For matrix $M$, $\abs{M}$ is defined to be $\sqrt{M^* M}$.
The operator norm $\norm{M}$ of matrix $M$ is the largest eigenvalue
of $\abs{M}$.
The trace norm $\norm{M}_1$ of $M$ is the trace of $\abs{M}$.
For two operators $M,N\in\Lin(\H)$, the Hilbert-Schmidt inner product
is defined as
\begin{equation*}
  \ip{M}{N} = \tr(M^*N).
\end{equation*}


A quantum channel is a physically admissible transformation of quantum
states.
Mathematically, a quantum channel
\begin{equation*}
  \mathcal{E}:\Lin(\H) \rightarrow \Lin(\K)
\end{equation*}
is a completely positive, trace-preserving linear map.

The trace distance of two quantum states $\rho_0, \rho_1 \in \Density(\H)$
is
\begin{equation}
  \label{eq:dtr}
  \dtr(\rho_0, \rho_1) \defeq \frac{1}{2} \norm{\rho_0 - \rho_1}_1.
\end{equation}
It is well known that for a state drawn uniformly at random from the
set $\{\rho_0, \rho_1\}$, the optimal distinguish probability is given by
\begin{equation*}
  \frac{1 + \dtr(\rho_0, \rho_1)}{2}.
\end{equation*}

Define number $N = 2^n$ and set $\X = \{0,1,\ldots,N-1\}$.
The quantum Fourier transform on $n$ qubits is defined as
\begin{equation}
  \label{eq:Fourier}
  F = \frac{1}{\sqrt{N}} \sum_{x,y\in\X} \omega_N^{xy} \ket{x}\bra{y}.
\end{equation}
It is a well-known fact in quantum computing that $F$ can be
implemented in time $\poly(n)$.

For Hilbert space $\H$ and integer $m$, we use $\vee^m \H$ to denote the
symmetric subspace of $\H^{\otimes m}$, the subspace of states that are
invariant under permutations of the subsystems.
Let $N$ be the dimension of $\H$ and let $\X$ be the set
$\{0,1,\ldots,N-1\}$ such that $\H$ is the span of $\{ \ket{x}
\}_{x\in\X}$.
For any $\x=(x_1,x_2,\ldots,x_m)\in\X^m$, define state
\begin{equation}
  \label{eq:symmetric-state-1}
  \bigket{\x;\Sym} = \sqrt{\frac{\prod_{j\in\X} m_j!}{m!}} \sum_{\sigma}
  \Bigket{x_{\sigma(1)}, x_{\sigma(2)}, \ldots, x_{\sigma(m)}}.
\end{equation}
The summation runs over all possible permutations $\sigma$ that give
different tuples $(x_{\sigma(1)}, x_{\sigma(2)}, \ldots, x_{\sigma(m)})$.
Equivalently, we have
\begin{equation}
  \label{eq:symmetric-state-2}
  \bigket{\x;\Sym} = \frac{1}{\sqrt{m! \prod_{j\in\X}m_j!}} \sum_{\sigma\in S_m}
  \Bigket{x_{\sigma(1)}, x_{\sigma(2)}, \ldots, x_{\sigma(m)}}.
\end{equation}
The coefficients in the front of the above two equations are
normalization constants.
The set of states
\begin{equation*}
  \bigl\{ \bigket{\x;\Sym} \bigr\}_{\x\in\X^m}
\end{equation*}
forms an orthonormal basis of the symmetric subspace $\vee^m \H$.
This implies that the dimension of the symmetric subspace is
\begin{equation*}
  \binom{N+m-1}{m}.
\end{equation*}

Let $\Pi^\Sym_m$ be the projection onto the symmetric subspace $\vee^m
\H$.
For a permutation $\sigma \in S_m$, define operator
\begin{equation*}
  W_\sigma = \sum_{x_1,x_2,\ldots,x_m \in \X} \bigket{x_{\sigma^{-1}(1)}, x_{\sigma^{-1}(2)},
    \ldots, x_{\sigma^{-1}(m)}} \bigbra{x_1, x_2, \ldots, x_m}.
\end{equation*}
We have the following useful identity
\begin{equation}
  \label{eq:symmetric-sum}
  \Pi^\Sym_m = \frac{1}{m!} \sum_{\sigma \in S_m} W_\sigma.
\end{equation}
Let $\mu$ be the Haar measure on $\State(\H)$, we have
\begin{equation*}
  \int \bigl( \ket{\psi}\bra{\psi} \bigr)^{\otimes m} \dif \mu(\psi) =
  \binom{N+m-1}{m}^{-1} \Pi^\Sym_m.
\end{equation*}


\subsection{Cryptography}

In this section, we recall several definitions and results from
cryptography that is necessary for this work.

Pseudorandom functions (PRF) and pseudorandom permutations (PRP) are
important constructions in classical cryptography.
Intuitively, they are families of functions or permutations that looks
like truly random functions or permutations to polynomial-time
machines.
In the quantum case, we need a strong requirement that they still look
random even to polynomial-time quantum algorithms.


\begin{definition}[Quantum-Secure Pseudorandom Functions and Permutations]
  Let $\K$, $\X$, $\Y$ be the key space, the domain and range, all
  implicitly depending on the security parameter $\secpar$.
  A keyed family of functions $\bigl\{ \prf_k : \X \rightarrow \Y
  \bigr\}_{k\in\K}$ is a quantum-secure pseudorandom function (QPRF) if
  for any polynomial-time quantum oracle algorithm $\A$, $\prf_k$ with
  a random $k \gets \K$ is indistinguishable from a truly random function
  $f \gets \Y^\X$ in the sense that:
  \begin{equation*}
    \abs{ \Pr_{k\gets\K} \bigl[ \A^{\prf_k}(\secparam) = 1 \bigr] -
      \Pr_{f\gets\Y^\X} \bigl[ \A^{f}(\secparam) = 1 \bigr]} =
    \negl(\secpar).
  \end{equation*}
  Similarly, a keyed family of permutations $\bigl\{ \prp_k \in S_\X
  \bigr\}_{k\in \K}$ is a quantum-secure pseudorandom permutation
  (QPRP) if for any quantum algorithm $\A$ making at most polynomially
  many queries, $\prp_k$ with a random $k\gets\K$ is indistinguishable
  from a truly random permutation in the sense that:
  \begin{equation*}
    \abs{\Pr_{k \gets \K} \bigl[ \A^{\prp_k}(\secparam) = 1 \bigr] -
      \Pr_{P \gets S_\X} \bigl[ \A^{P}(\secparam) = 1 \bigr]} =
    \negl(\secpar).
  \end{equation*}
  In addition, both $\prf_k$ and $\prp_k$ are polynomial-time
  computable (on a classical computer).
  \label{def:prfprp}
\end{definition}

\begin{fact}
  QPRFs and QPRPs exist if quantum-secure one-way functions exist.
\end{fact}

Zhandry proved the existence of QPRFs assuming the existence of
one-way functions that are hard to invert even for quantum
algorithms~\cite{Zha12}.
Assuming QPRF, one can construct QPRP using various \emph{shuffling}
constructions~\cite{Zha16,Song17}.
Since a random permutation and a random function is indistinguishable
by efficient quantum algorithms~\cite{Yuen14,Zha15}, existence of QPRP
is hence equivalent to existence of QPRF.


\section{Pseudorandom Quantum States}
\label{sec:prs}

In this section, we will discuss the definition and constructions of
pseudorandom quantum states.

\subsection{Definition of Pseudorandom States}
\label{sec:definition}


Intuitively speaking, a family pseudorandom quantum states are a set
of random states $\bigl\{ \ket{\phi_k} \bigr\}_{k\in\K}$ that is
indistinguishable from Haar random quantum states.

The first idea on defining pseudorandom states can be the following.
Without loss of generality, we consider states in $\State(\H)$ where
$\H = (\complex^2)^{\otimes n}$ is the Hilbert space for $n$-qubit systems.
We are given either a state randomly sampled from the set $\bigl\{
\ket{\phi_k} \in \H \bigr\}_{k\in\K}$ or a state sampled according to the
Haar measure on $\State(\H)$, and we require that no efficient quantum
algorithm will be able to tell the difference between the two case.

However, this definition does not seem to grasp the quantum nature of
the problem.
First, the state family where each $\ket{\phi_k}$ is a uniform random
bit string will satisfy the definition---in both cases, the mixed states
representing the ensemble are $\I/2^n$.
Second, many of the applications that we can find for PRS's will not
hold for this definition.

Instead, we require that the family of states looks random even if
polynomially many copies of the state are given to the distinguishing
algorithm.
We argue that this is the more natural way to define pseudorandom
states.
One can see that this definition also naturally generalizes the
definition of pseudorandomness in the classical case to the quantum
setting.
In the classical case, asking for more copies of a string is always
possible and one does not bother making this explicit in the
definition.
This of course also rules out the example of classical random bit
strings we discussed before.
Moreover, this strong definition, once established, is rather flexible
to use when studying the properties and applications of pseudorandom
states.

\begin{definition}[Pseudorandom Quantum States (PRS's)]
  \label{def:prs}
  Let $\H$ be a Hilbert space and $\K$ the key space.
  $\H$ and $\K$ depend on the security parameter $\secpar$.
  A keyed family of quantum states $\bigl\{ \ket{\phi_k} \in \State(\H)
  \bigr\}_{k\in\K}$ is \textbf{pseudorandom}, if the following two
  conditions hold:
  \begin{enumerate}
  \item (\textbf{Efficient generation}).
    There is a polynomial-time quantum algorithm $G$ that generates
    state $\ket{\phi_k}$ on input $k$.
    That is, for all $k\in\K$, $G(k) = \ket{\phi_k}$.
  \item (\textbf{Pseudorandomness}).
    Any polynomially many copies of $\ket{\phi_k}$ with the same random
    $k\in\K$ is \textbf{computationally indistinguishable} from the
    same number of copies of a Haar random state.
    More precisely, for any efficient quantum algorithm $\A$ and any
    $m\in\poly(\secpar)$,
    \begin{equation*}
      \abs{\Pr_{k\gets\K} \bigl[ \A(\ket{\phi_k}^{\otimes m}) = 1 \bigr] -
        \Pr_{\ket{\psi}\gets \mu} \bigl[ \A(\ket{\psi}^{\otimes m}) = 1 \bigr] } =
      \negl(\secpar),
    \end{equation*}
    where $\mu$ is the Haar measure on $\State(\H)$.
  \end{enumerate}
\end{definition}

\subsection{Constructions and Analysis}
\label{sec:construction}

In this section, we give an efficient construction of pseudorandom
states which we call random phase states.
We will prove that this family of states satisfies our definition of
PRS's.
There are other interesting and simpler candidate constructions, but
the family of random phase states is the easiest to analyze.


Let $\prf:\K \times \X \rightarrow \X$ be a quantum-secure pseudorandom function
with key space $\K$, $\X = \{0, 1, 2, \ldots, N-1\}$ and $N=2^n$.
$\K$ and $N$ are implicitly functions of the security parameter
$\secpar$.
The family of pseudorandom states of $n$ qubits is defined
\begin{equation}
  \label{eq:random-phase-states}
  \ket{\phi_k} = \frac{1}{\sqrt{N}} \sum_{x\in \X} \omega_N^{\prf_k(x)}
  \ket{x},
\end{equation}
for $k \in \K$ and $\omega_N = \exp(2\pi i/N)$.


\begin{theorem}
  For any QPRF $\prf:\K \times \X \rightarrow \X$, the family of states
  $\{\ket{\phi_k}\}_{k\in\K}$ defined in
  Eq.~\eqref{eq:random-phase-states} is a PRS.
\end{theorem}

\begin{proof}
  First, we prove that the state can be efficiently prepared with a
  single query to $\prf_k$.
  As $\prf_k$ is efficient, this proves the efficient generation
  property.

  The state generation algorithm works as follows.
  First, it prepares a state
  \begin{equation*}
    \frac{1}{N} \sum_{x\in\X} \ket{x} \sum_{y\in\X} \omega_N^y \ket{y}.
  \end{equation*}
  This can be done by applying $H^{\otimes n}$ to the first register
  initialized in $\ket{0}$ and the quantum Fourier transform to the
  second register in state $\ket{1}$.

  Then the algorithm calls $\prf_k$ on the first register and subtract
  the result from the second register, giving state
  \begin{equation*}
    \frac{1}{N} \sum_{x\in\X} \ket{x} \sum_{y\in\X} \omega_N^y \bigket{y - \prf_k(x)}.
  \end{equation*}
  The state can be rewritten as
  \begin{equation*}
    \frac{1}{N} \sum_{x\in\X} \omega_N^{\prf_k(x)} \ket{x} \sum_{y\in\X} \omega_N^y \ket{y}.
  \end{equation*}
  Therefore, the effect of this step is to transform the first
  register to the required form and leaving the second register
  intact.

  Next, we prove the pseudorandomness property of the family.
  For this purpose, we consider three hybrids.
  In the first hybrid $H_1$, the state will be $\ket{\phi_k}^{\otimes m}$ for
  a uniform random $k\in\K$.
  In the second hybrid $H_2$, the state is $\ket{f}^{\otimes m}$ for truly
  random functions $f\in\X^\X$ where
  \begin{equation*}
    \ket{f} = \frac{1}{\sqrt{N}} \sum_{x\in \X} \omega_N^{f(x)} \ket{x}.
  \end{equation*}
  In the third hybrid $H_3$, the state is $\ket{\psi}^{\otimes m}$ for
  $\ket{\psi}$ chosen according to the Haar measure.

  By the definition of the quantum-secure pseudorandom functions for
  $\prf$, we have for any polynomial-time quantum algorithm $\A$ and
  any $m\in\poly(\secpar)$,
  \begin{equation*}
    \abs{\Pr\bigl[ \A(H_1) = 1 \bigr] - \Pr\bigl[ \A(H_2) = 1 \bigr]} =
    \negl(\secpar).
  \end{equation*}

  By Lemma~\ref{lem:random-phase}, we have for any algorithm $\A$ and
  $m\in\poly(\secpar)$,
  \begin{equation*}
    \abs{\Pr\bigl[ \A(H_2) = 1 \bigr] - \Pr\bigl[ \A(H_3) = 1 \bigr]} =
    \negl(\secpar).
  \end{equation*}

  This completes the proof by triangle inequality.
\end{proof}


\begin{lemma}
  \label{lem:random-phase}
  For function $f:\X \rightarrow \X$, define quantum state
  \begin{equation*}
    \ket{f} = \frac{1}{\sqrt{N}} \sum_{x\in \X} \omega_N^{f(x)} \ket{x}.
  \end{equation*}
  For $m\in\poly(\secpar)$, the state ensemble $\bigl\{ \ket{f}^{\otimes m}
  \bigr\}$ is statistically indistinguishable from $\bigl\{
  \ket{\psi}^{\otimes m} \bigr\}$ for Haar random $\ket{\psi}$.
\end{lemma}

\begin{proof}
  Let $m\in\poly(\secpar)$ be the number of copies of the state.
  We have
  \begin{equation*}
    \E_f \Bigl[ \bigl( \ket{f}\bra{f} \bigr)^{\otimes m} \Bigr] =
    \frac{1}{N^m} \sum_{\x \in \X^m, \y\in \X^m} \E_f
    \omega_N^{f(x_1) + \cdots + f(x_m) - [f(y_1) + \cdots + f(y_m)]}
    \bigket{\x}\bigbra{\y},
  \end{equation*}
  where $\x = (x_1, x_2, \ldots, x_m)$ and $\y = (y_1, y_2, \ldots, y_m)$.
  For later convenience, define density matrix
  \begin{equation*}
    \rho^m = \E_f \Bigl[ \bigl( \ket{f}\bra{f} \bigr)^{\otimes m} \Bigr].
  \end{equation*}
  We will compute the entries of $\rho^m$ explicitly.

  For $\x = (x_1, x_2, \ldots, x_m) \in \X^m$, let $m_j$ be the number of
  $j$ in $\x$ for $j\in\X$.
  Obviously, one has $\sum_{j\in\X} m_j = m$.
  Note that we have omitted the dependence of $m_j$ on $\x$ for
  simplicity.
  Recall the basis states defined in Eq.~\eqref{eq:symmetric-state-2}
  \begin{equation*}
    \bigket{\x;\Sym} =
    \frac{1}{\sqrt{\Bigl(\prod_{j\in\X}m_j! \Bigr) m!}} \sum_{\sigma\in S_m}
    \Bigket{x_{\sigma(1)}, x_{\sigma(2)}, \ldots, x_{\sigma(m)}}.
  \end{equation*}
  For $\x, \y \in \X^m$, let $m_j$ be the number of $j$ in $\x$ and
  $m_j'$ be the number of $j$ in $\y$.

  We can compute the entries of $\rho^m$ as
  \begin{equation*}
    \bigbra{\x;\Sym} \rho^m
    \bigket{\y;\Sym} = \frac{m!}{N^m \sqrt{\Bigl(\prod_{j\in\X} m_j!
        \Bigr) \Bigl(\prod_{j\in\X} m_j'! \Bigr)}} \E_f \biggl[ \exp\Bigl(
    \frac{2\pi i}{N} \sum_{l=1}^m \bigr( f(x_l) - f(y_l) \bigr) \Bigr)
    \biggr].
  \end{equation*}
  It is not hard to verify that the entry is nonzero only if $\x$ is a
  permutation of $\y$.
  These nonzero entries are on the diagonal of $\rho^m$ in the basis of
  $\bigl\{ \bigket{\x;\Sym} \bigr\}$.
  These diagonal entries are
  \begin{equation*}
    \bigbra{\x;\Sym} \rho^m
    \bigket{\x;\Sym} = \frac{m!}{N^m
      \prod_{j\in\X} m_j!}.
  \end{equation*}

  Let $\rho^m_\mu$ be the density matrix of a random $\ket{\psi}^{\otimes m}$ for
  $\ket{\psi}$ chosen from the Haar measure $\mu$.
  It is well-known that
  \begin{equation*}
    \rho^m_\mu = \binom{N+m-1}{m}^{-1} \sum_{\x;\Sym} \bigket{\x;\Sym}
    \bigbra{\x;\Sym}
  \end{equation*}

  We need to prove
  \begin{equation*}
    \dtr \bigl( \rho^m, \rho^m_\mu \bigr) = \negl(\secpar).
  \end{equation*}

  Define
  \begin{equation*}
    \delta_{\x;\Sym} = \frac{m!}{N^m \prod_{j\in\X} m_j!} -
    \binom{N+m-1}{m}^{-1}.
  \end{equation*}
  Then
  \begin{equation*}
    \dtr(\rho^m, \rho^m_\mu) = \frac{1}{2} \sum_{\x;\Sym} \abs{\delta_{\x;\Sym}}.
  \end{equation*}

  The ratio of the two terms in $\delta_{\x;\Sym}$ is
  \begin{equation*}
    \frac{\displaystyle m! \binom{N+m-1}{m}}{\displaystyle N^m
      \prod_{j\in\X} m_j!} = \frac{\displaystyle \prod_{l=0}^{m-1} \Bigl(
      1+\frac{l}{N} \Bigr)}{\displaystyle \prod_{j\in\X} m_j!}.
  \end{equation*}
  For sufficient large security parameter $\secpar$, the ratio is
  larger than $1$ only if $\prod_{j\in\X}m_j! = 1$, which corresponds to
  $\x$'s whose entries are all distinct.
  As there are $\binom{N}{m}$ such $\x$'s, we can calculate the trace
  distance as
  \begin{equation*}
    \begin{split}
      \dtr \bigl(\rho^m, \rho^m_\mu \bigr) & = \binom{N}{m} \biggl[
      \frac{m!}{N^m} -
      \binom{N+m-1}{m}^{-1} \biggr]\\
      & = \frac{N (N-1) \cdots (N-m+1)}{N^m} - \frac{N (N-1) \cdots
        (N-m+1)}{(N+m-1) \cdots N}.
    \end{split}
  \end{equation*}
  For $m\in\poly(\secpar)$, both terms in the last line of the equation
  is $1-\negl(\secpar)$ for sufficiently large security parameter
  $\secpar$, and this completes the proof.
\end{proof}

We remark that a similar family of states was considered
in~\cite{Aar09} (Theorem 3).
However, the size of the state family there depends on a parameter $d$
which should be larger than the sum of the number of state copies and
the number of queries.
In our construction, the key space is fixed for a given security
parameter, which may be advantageous for various applications.

We mention several other candidate constructions of PRS's and leave
detailed analysis of them to future work.
A construction closely related to the random phase states in
Eq.~\eqref{eq:random-phase-states} uses random $\pm 1$ phases,
\begin{equation*}
  \ket{\phi_k} = \frac{1}{\sqrt{N}} \sum_{x\in\X} (-1)^{\prf_k(x)} \ket{x}.
\end{equation*}
Intuitively, this family is less random than the random phase states
in Eq.~\eqref{eq:random-phase-states} and the corresponding density
matrix $\rho^m$ has small off-diagonal entries, making the proof more
challenging.
The other family of candidate states on $2n$ qubits takes the form
\begin{equation*}
  \ket{\phi_k} = \frac{1}{\sqrt{N}} \prp_k \biggl[ \sum_{x\in\X}
  \ket{x}\otimes\ket{0^n} \biggr].
\end{equation*}
In this construction, the state is an equal superposition of a random
subset of size $2^n$ of $\{0,1\}^{2n}$ and $\prp$ is any pseudorandom
permutation over the set $\{0,1\}^{2n}$.
We call this the \emph{random subset states} construction.

Our PRS constructions can be implemented using shallow quantum
circuits of polylogarithmic depth under appropriate cryptographic
assumptions.
To see this, note that there exist PRFs that can be computed in
polylogarithmic depth, which are based on lattice problems such as
``learning with errors'' (LWE)~\cite{Reg09}, and are believed to be
secure against quantum computers~\cite{BPR12}.
These PRFs can be used directly in our PRS construction.
(Alternatively, one can use low-depth PRFs that are constructed from
more general assumptions, such as the existence of trapdoor one-way
permutations~\cite{Nao99}.)

This shows that PRS states can be prepared in surprisingly small
depth, compared to quantum state $t$-designs, which generally require
at least linear depth (when $t$ is a constant greater than 2), or
polynomial depth (when $t$ grows polynomially with the number of
qubits)~\cite{AE07,BHH16}.
(Note, however, that for $t=2$, quantum state 2-designs can be
generated in logarithmic depth~\cite{CLLW16}.)
Moreover, PRS states are sufficient for many applications where
high-order $t$-designs are used, provided that one only requires
states to be \textit{computationally} (not statistically)
indistinguishable from Haar-random.

\section{Cryptographic Non-cloning Theorem and Quantum Money}
\label{sec:non-cloning}


A fundamental fact in quantum information theory is that unknown or
random quantum states cannot be
cloned~\cite{WZ82,Die82,Wer98,Ort17,Par70}.
The main topic of this section is to investigate the cloning problem
for pseudorandom states.
As we will see, even though pseudorandom states can be efficiently
generated, they do share the non-cloning property of generic quantum
states.

Let $\H$ be the Hilbert space of dimension $N$ and $m<m'$ be two
integers.
The numbers $N,m,m'$ depend implicitly on a security parameter
$\secpar$.
We will assume that $N$ is exponential in $\secpar$ and
$m\in\poly(\secpar)$ in the following discussion.


We first recall the fact that for Haar random state $\ket{\psi} \in
\State(\H)$, the success probability of producing $m'$ copies of the
state given $m$ copies is negligibly small.
Let $\C$ be a cloning channel that on input $(\ket{\psi}\bra{\psi})^{\otimes m}$
tries to output a state that is close to $(\ket{\psi}\bra{\psi})^{\otimes m'}$
for $m'>m$.
The expected success probability of $\C$ is measured by
\begin{equation*}
  \int \Bigl\langle \bigl( \ket{\psi}\bra{\psi} \bigr)^{\otimes m'}, \C \bigl( \bigl(
  \ket{\psi}\bra{\psi} \bigr)^{\otimes m} \bigr) \Bigr\rangle \dif \mu(\psi).
\end{equation*}
It is known that~\cite{Wer98}, for all cloning channel $\C$, this
success probability is bounded by
\begin{equation*}
  \binom{N+m-1}{m} \bigg/ \binom{N+m'-1}{m'},
\end{equation*}
which is $\negl(\secpar)$ for our choices of $N,m,m'$.


We establish a non-cloning theorem for PRS's which says that no
efficient quantum cloning procedure exists for a general PRS.
The theorem is called the cryptographic non-cloning theorem because of
its deep roots in pseudorandomness in cryptography.

\begin{theorem}[Cryptographic Non-cloning Theorem]
  \label{thm:non-cloning-simple}
  For any PRS $\{ \ket{\phi_k} \}_{k\in\K}$, $m\in\poly(\secpar)$, $m<m'$
  and any polynomial-time quantum algorithm $\C$, the success cloning
  probability
  \begin{equation*}
    \E_{k\in\K} \Bigl\langle \bigl( \ket{\phi_k}\bra{\phi_k} \bigr)^{\otimes m'}, \C
    \bigl( \bigl(\ket{\phi_k}\bra{\phi_k} \bigr)^{\otimes m} \bigr) \Bigr\rangle =
    \negl(\secpar).
  \end{equation*}
\end{theorem}

\begin{proof}
  Assume on the contrary that there is a polynomial-time quantum
  cloning algorithm $\C$ such that the success cloning probability of
  producing $m+1$ from $m$ copies is $\secpar^{-c}$ for some constant
  $c>0$.
  We will construct a polynomial-time distinguisher $\D$ that violates
  the definition of PRS's.
  Distinguisher $\D$ will draw $2m+1$ copies of the state, call $\C$
  on the first $m$ copies, and perform the SWAP test on the output of
  $\C$ and the remaining $m+1$ copies.
  It is easy to see that $\D$ outputs $1$ with probability
  $(1+\secpar^{-c})/2$ if the input is from PRS, while if the input is
  Haar random, it outputs $1$ with probability $(1+\negl(\secpar))/2$.
  Since $\C$ is polynomial-time, it follows that $\D$ is also
  polynomial-time.
  This is a contradiction with the definition of PRS's and completes
  the proof.
\end{proof}


\subsection{A Strong notion of PRS and equivalence to PRS}
\label{sec:strong-prs}

In this section, we show that, somewhat surprisingly, PRS in fact
implies a seemingly stronger notion, where indistinguishability needs
to hold even if a distinguisher additionally has access to an oracle
that reflects about the given state.
There are at least a couple of motivations to consider an augmented
notion.
Firstly, unlike a classical string, a quantum state is inherently
\emph{hidden}.
Give a quantum register prepared in some state (i.e., a physical
system), we can only choose some observable to measure which just
reveals partial information and will collapse the state in general.
Therefore, it is meaningful to consider offering a distinguishing
algorithm more information \emph{describing} the given state, and the
reflection oracle comes naturally.
Secondly, this stronger notion is extremely useful in our application
of quantum money schemes, and could be interesting elsewhere too.

More formally, for any state $\ket{\phi}\in\H$, define an oracle $O_\phi:=
\I - 2\ket{\phi}\bra{\phi}$ that reflects about $\ket{\phi}$.



\begin{definition}[Strongly Pseudorandom Quantum States]
  \label{def:strong-prs}
  Let $\H$ be a Hilbert space and $\K$ be the key space.
  $\H$ and $\K$ depend on the security parameter $\secpar$.
  A keyed family of quantum states $\bigl\{ \ket{\phi_k} \in \State(\H)
  \bigr\}_{k\in\K}$ is \textbf{strongly pseudorandom}, if the following
  two conditions hold:
  \begin{enumerate}
  \item (\textbf{Efficient generation}).
    There is a polynomial-time quantum algorithm $G$ that generates
    state $\ket{\phi_k}$ on input $k$.
    That is, for all $k\in\K$, $G(k) = \ket{\phi_k}$.
  \item (\textbf{Strong Pseudorandomness}).
    Any polynomially many copies of $\ket{\phi_k}$ with the same random
    $k\in\K$ is \textbf{computationally indistinguishable} from the
    same number of copies of a Haar random state.
    More precisely, for any efficient quantum oracle algorithm $\A$
    and any $m\in\poly(\secpar)$,
    \begin{equation*}
      \abs{\Pr_{k\gets\K} \bigl[ \A^{O_{\phi_k}}(\ket{\phi_k}^{\otimes m}) = 1 \bigr] -
        \Pr_{\ket{\psi}\gets \mu} \bigl[ \A^{O_\psi}(\ket{\psi}^{\otimes m}) = 1 \bigr] } =
      \negl(\secpar),
    \end{equation*}
    where $\mu$ is the Haar measure on $\State(\H)$.
  \end{enumerate}
\end{definition}

Note that since the distinguisher $\A$ is polynomial-time, the number
of queries to the reflection oracle ($O_{\phi_k}$ or $O_\psi$) is also
polynomially bounded.


We prove the advantage that a reflection oracle may give to a
distinguisher is limited.
In fact, standard PRS implies strong PRS, and hence they are
equivalent.

\begin{theorem}
  \label{thm:equivalence}
  A family of states $\bigl\{ \ket{\phi_k} \bigr\}_{k\in \K}$ is strongly
  pseudorandom if and only if it is (standard) pseudorandom.
\end{theorem}

\begin{proof} Clearly a strong PRS is also a standard PRS by
  definition.
  It suffice to prove that any PRS is also strongly pseudorandom.

  Suppose for contradiction that there is a distinguishing algorithm
  $\A$ that breaks the strongly pseudorandom condition.
  Namely, there exists $m\in\poly(\secpar)$ and constant $c>0$ such
  that for sufficiently large $\secpar$,
  \begin{equation*}
    \abs{ \Pr_{k\gets\K} \bigl[ \A^{O_{\phi_k}} (\ket{\phi_k}^{\otimes m}) = 1 \bigr] -
      \Pr_{\ket{\psi}\gets\mu} \bigl[ \A^{O_\psi}(\ket{\psi}^{\otimes m}) = 1 \bigr]} =
    \varepsilon(\secpar) \geq \secpar^{-c}.
  \end{equation*}

  We assume $\A$ makes $q\in\poly(\secpar)$ queries to the reflection
  oracle.
  Then, by Theorem~\ref{thm:simulation}, there is an algorithm $\B$
  such that for any $l$
  \begin{equation*}
    \abs{ \Pr_{k\gets\K} \bigl[ \A^{O_{\phi_k}}(\ket{\phi_k}^{\otimes m}) \bigr] -
      \Pr_{k\gets\K} \bigl[ \B(\ket{\phi_k}^{\otimes (m+l)}) \bigr] } \le
    \frac{2q}{\sqrt{l+1}},
  \end{equation*}
  and
  \begin{equation*}
    \abs{ \Pr_{\ket{\psi}\gets\mu} \bigl[ \A^{O_\psi}(\ket{\psi}^{\otimes m}) \bigr] -
      \Pr_{\ket{\psi}\gets\mu} \bigl[ \B(\ket{\psi}^{\otimes (m+l)}) \bigr] } \le
    \frac{2q}{\sqrt{l+1}}.
  \end{equation*}

  By triangle inequality, we have
  \begin{equation*}
    \abs{ \Pr_{k\gets\K} \bigl[ \B(\ket{\phi_k}^{\otimes (m+l)}) \bigr] -
      \Pr_{\ket{\psi}\gets\mu} \bigl[ \B(\ket{\psi}^{\otimes (m+l)}) \bigr]} \ge
    \secpar^{-c} - \frac{4q}{\sqrt{l+1}}.
  \end{equation*}
  Choosing $l = 64q^2\secpar^{2c} \in \poly(\secpar)$, we have
  \begin{equation*}
    \abs{ \Pr_{k\gets\K} \bigl[ \B(\ket{\phi_k}^{\otimes (m+l)}) \bigr] -
      \Pr_{\ket{\psi}\gets\mu} \bigl[ \B(\ket{\psi}^{\otimes (m+l)}) \bigr]} \ge
    \secpar^{-c}/2,
  \end{equation*}
  which is a contradiction with the definition of PRS for
  $\{\ket{\phi_k}\}$.
  Therefore, we conclude that PRS and strong PRS are equivalent.
\end{proof}


We now show a technical ingredient that allows us to simulate the
reflection oracle about a state by using multiple copies of the given
state.
This result is inspired by a similar theorem proved by Ambainis et
al.~\cite[Lemma 42]{ARU14}.
We have improved their reduction by replacing a reflection operation
about a particular subspace to the reflection about the standard
symmetric subspace, which we know how to implement efficiently.

\begin{theorem}
  \label{thm:simulation}
  Let $\ket{\psi}\in\H$ be a quantum state.
  Define oracle $O_\psi = \I - 2\ket{\psi}\bra{\psi}$ to be the reflection
  about $\ket{\psi}$.
  Let $\ket{\xi}$ be a state not necessarily independent of $\ket{\psi}$.
  Let $\A^{O_\psi}$ be an oracle algorithm that makes $q$ queries to
  $O_\psi$.
  For any integer $l>0$, there is a quantum algorithm $\B$ that does
  not make any queries to $O_\psi$ such that
  \begin{equation*}
    \dtr \bigl( \A^{O_\psi} (\ket{\xi}), \B (\ket{\psi}^{\otimes l} \otimes \ket{\xi})
    \bigr) \leq \frac{2q}{\sqrt{l+1}}.
  \end{equation*}
  Moreover, the running time of algorithm $\B$ is polynomial in that
  of algorithm $\A$ and $l$.
\end{theorem}

\begin{proof}
  Consider a quantum register \reg{T}, initialized in state
  $\ket{\Theta}_{\reg{T}} = \ket{\psi}^{\otimes l}$.
  Let $\Pi$ be the projection onto the symmetric subspace $\vee^{l+1}\H$
  and $R = \I- 2\Pi$ be the reflection about the symmetric subspace.

  Assume without loss of generality that algorithm $\A$ is unitary and
  only perform measurements in the end.
  We define algorithm $\B$ to be same as $\A$ except that when $\A$
  queries $O_\psi$ on register $\reg{D}$, $\B$ applies the reflection
  $R$ on the collection of quantum registers $\reg{D}$ and $\reg{T}$
  where $\reg{T}$ is initialized in state $\ket{\Theta}$.
  We first analyze the corresponding states after the first oracle
  call to $O_\psi$ in algorithms $\A$ and $\B$,
  \begin{equation*}
    \ket{\Psi_A} = O_\psi \ket{\phi}_{\reg{D}} \otimes \ket{\Theta}_{\reg{T}},\quad
    \ket{\Psi_B} = R \bigl( \ket{\phi}_{\reg{D}} \otimes \ket{\Theta}_{\reg{T}} \bigr).
  \end{equation*}

  For any two state $\ket{x}, \ket{y}\in\H$, we have
  \begin{equation*}
    \begin{split}
      \bra{x, \Theta} R \ket{y, \Theta} & = \braket{x}{y} - 2 \E_{\pi \in
        S_{l+1}} (\bra{x} \otimes \bra{\Theta}) W_\pi (\ket{y} \otimes \ket{\Theta})\\
      & = \braket{x}{y} - \frac{2}{l+1} \braket{x}{y} - \frac{2l}{l+1}
      \braket{x}{\psi} \braket{\psi}{y}\\
      & = \frac{l-1}{l+1} \braket{x}{y} - \frac{2l}{l+1}
      \braket{x}{\psi} \braket{\psi}{y},
    \end{split}
  \end{equation*}
  where the first step uses the identity in
  Eq.~\eqref{eq:symmetric-sum} and the second step follows by
  observing that the probability of a random $\pi\in S_{l+1}$ mapping
  $1$ to $1$ is $1/(l+1)$.
  These calculations imply that,
  \begin{equation*}
    \bra{\Theta} R \ket{\Theta} = \frac{l-1}{l+1}\I - \frac{2l}{l+1}
    \ket{\psi}\bra{\psi}.
  \end{equation*}

  We can compute the inner product of the two states $\ket{\Psi_A}$ and
  $\ket{\Psi_B}$ as
  \begin{equation*}
    \begin{split}
      \braket{\Psi_A}{\Psi_B} & =
      \ip{\ket{\phi,\Theta}\bra{\phi,\Theta}}{(O_\psi\otimes\I)R}\\
      & = \ip{\ket{\phi}\bra{\phi}}{O_\psi \bra{\Theta}R\ket{\Theta}}\\
      & \ge 1 - \frac{2}{l+1}.
    \end{split}
  \end{equation*}
  This implies that
  \begin{equation*}
    \norm{\ket{\Psi_A} - \ket{\Psi_B}} \le \frac{2}{\sqrt{l+1}}.
  \end{equation*}

  Let $\ket{\Psi^q_A}$ and $\ket{\Psi^q_B}$ be the final states of
  algorithm $\A$ and $\B$ before measurement respectively.

  Then by induction on the number of queries, we have
  \begin{equation*}
    \norm{\ket{\Psi^q_A} - \ket{\Psi^q_B}} \le \frac{2q}{\sqrt{l+1}}.
  \end{equation*}
  This concludes the proof by noticing that
  \begin{equation*}
    \dtr \bigl( \ket{\Psi^q_A}, \ket{\Psi^q_B} \bigr) \le
    \norm{\ket{\Psi^q_A} - \ket{\Psi^q_B}}.
  \end{equation*}

  Finally, we show that if $\A$ is polynomial-time, then so is $\B$.
  Based on the construction of $\B$, it suffices to show that the
  reflection $R$ is efficiently implementable for any polynomially
  large $l$.
  Here we use a result by Barenco et al.~\cite{BBD+97} which proposes
  an efficient implementation for the projection $\Pi$ onto $\vee^{l+1}
  \H$.
  More precisely, they design a quantum circuit of size $O(\poly(l,
  \log \dim \H))$ that implements a unitary $U$ such that $U \ket{\phi}
  = \sum_{j}\ket{\xi_j}\ket{j}$ on $\H^{\otimes (l+1)}\otimes \H'$ for an auxiliary
  space $\H'$ of dimension $O(\dim(\H^{\otimes l}))$.
  Here $\ket{\xi_0} = \Pi \ket{\phi}$ corresponds to the projection of
  $\ket{\phi}$ on the symmetric subspace.
  With $U$, we can implement the reflection $R$ as $U^* S U$ where $S$
  is the unitary that introduces a minus sign conditioned on the
  second register being $0$.
  \begin{equation*}
    S \ket{\Psi} \ket{j} =
    \begin{cases}
      -\ket{\Psi} \ket{j} & \text{ if $j=0$,}\\
      \phantom{-}\ket{\Psi} \ket{j} & \text{ otherwise.}
    \end{cases}
  \end{equation*}
\end{proof}

\subsection{Quantum Money from PRS}
\label{sec:money}


Using Theorem~\ref{thm:equivalence}, we can improve
Theorem~\ref{thm:non-cloning-simple} to the following version.
The proof is omitted as it is very similar to that for
Theorem~\ref{thm:non-cloning-simple} and uses the complexity-theoretic
non-cloning theorem~\cite{Aar09,AC12} for Haar random states.


\begin{theorem}[Cryptographic Non-cloning Theorem with Oracle]
  \label{thm:non-cloning-oracle}
  For any PRS $\{ \ket{\phi_k} \}_{k\in\K}$, $m\in\poly(\secpar)$, $m<m'$
  and any polynomial-time quantum query algorithm $\C$, the success
  cloning probability
  \begin{equation*}
    \E_{k\in\K} \Bigl\langle \bigl( \ket{\phi_k}\bra{\phi_k} \bigr)^{\otimes m'},
    \C^{O_{\phi_k}} \bigl( \bigl(\ket{\phi_k}\bra{\phi_k}
    \bigr)^{\otimes m} \bigr) \Bigr\rangle = \negl(\secpar).
  \end{equation*}
\end{theorem}


A direct application of this non-cloning theorem is that it gives rise
to new constructions for private-key quantum money.
As one of the earliest foundational findings in quantum
information~\cite{Wie70,BBBW83}, quantum money schemes have received
revived interests in the past decade (see
e.g.~\cite{Aar09,LAF+10,MS10,FGH+10,FGH+12,AFG12}).
First, we recall the definition of quantum money scheme adapted
from~\cite{AC12}.

\begin{definition}[Quantum Money Scheme]
  A private-key quantum money scheme $\S$ consists of three
  algorithms:
  \begin{itemize}
  \item \KeyGen, which takes as input the security parameter
    $\secparam$ and randomly samples a private key $k$.
  \item \Bank, which takes as input the private key $k$ and generates
    a quantum state $\ket{\money}$ called a \textbf{banknote}.
  \item \Ver, which takes as input the private key $k$ and an alleged
    banknote $\ket{\cent}$, and either accepts or rejects.
  \end{itemize}

  The money scheme $\S$ has \textbf{completeness error} $\varepsilon$ if
  $\Ver\,(k,\ket{\money})$ accepts with probability at least $1-\varepsilon$
  for all valid banknote $\ket{\money}$.

  Let \Count be the money counter that output the number of valid
  banknotes when given a collection of (possibly entangled) alleged
  banknotes $\ket{\cent_1, \cent_2, \ldots, \cent_r}$.
  Namely, \Count will call \Ver on each banknotes and return the
  number of times that \Ver accepts.
  The money scheme $\S$ has \textbf{soundness error} $\delta$ if for any
  polynomial-time counterfeiter $C$ that maps $q$ valid banknotes
  $\ket{\money_1}, \ldots, \ket{\money_q}$ to $r$ alleged banknotes
  $\ket{\cent_1, \ldots, \cent_r}$ satisfies
  \begin{equation*}
    \Pr \bigl[ \Count\, \bigl( k, C(\ket{\money_1}, \ldots, \ket{\money_q})
    \bigr) > q \bigr] \le \delta.
  \end{equation*}
  The scheme $\S$ is \textbf{secure} if it has completeness error $\le
  1/3$ and negligible soundness error.
\end{definition}


For any $\prs = \bigl\{ \ket{\phi_k} \bigr\}_{k\in\K}$ with key space
$\K$, we can define a private-key quantum money scheme $\S_\prs$ as
follows:
\begin{itemize}
\item $\KeyGen(\secparam)$ randomly outputs $k\in\K$.
\item $\Bank(k)$ generates the banknote $\ket{\money} = \ket{\phi_k}$.
\item $\Ver(k,\rho)$ applies the projective measurement that accepts
  $\rho$ with probability $\bra{\phi_k} \rho \ket{\phi_k}$.
\end{itemize}

We remark that usually the money state $\ket{\money}$ takes the form
$\ket{\money} = \ket{s, \psi_s}$ where the first register contains a
classical serial number.
Our scheme, however, does not require the use of the serial numbers.
This simplification is brought to us by the strong requirement that
any polynomial copies of $\ket{\phi_k}$ are indistinguishable from Haar
random states.


\begin{theorem}
  The private-key quantum money scheme $\S_\prs$ is secure for all
  $\prs$.
\end{theorem}

\begin{proof}
  It suffices to prove the soundness of $S_\prs$ is negligible.
  Assume to the contrary that there is a counterfeiter $C$ such that
  \begin{equation*}
    \Pr \bigl[ \Count \bigl( k, C(\ket{\phi_k}^{\otimes q}) \bigr) > q \bigr]
    \ge \secpar^{-c}
  \end{equation*}
  for some constant $c>0$ and sufficiently large $\secpar$.
  From the counterfeiter $C$, we will construct an oracle algorithm
  $\A^{O_{\phi_k}}$ that maps $q$ copies of $\ket{\phi_k}$ to $q+1$ copies
  with noticeable probability and this leads to a contradiction with
  Theorem~\ref{thm:non-cloning-oracle}.

  The oracle algorithm $\A$ first runs $C$ and implement the
  measurement
  \begin{equation*}
    \Bigl\{ M^0 = \I - \ket{\phi_k}\bra{\phi_k}, M^1 =
    \ket{\phi_k}\bra{\phi_k} \Bigr\}
  \end{equation*}
  on each copy of the money state $C$ outputs.
  This measurement can be implemented by attaching an auxiliary qubit
  initialized in $(\ket{0} +\ket{1})/\sqrt{2}$ and call the reflection
  oracle $O_\phi$ conditioned on the qubit being at $1$ and performs the
  $X$ measurement on this auxiliary qubit.
  This gives $r$-bit of outcome $\x\in\{0,1\}^r$.
  If $\x$ has Hamming weight at least $q+1$, algorithm $\A$ outputs
  any $q+1$ registers that corresponds to outcome $1$; otherwise, it
  outputs $\ket{0}^{\otimes (q+1)}$.
  By the construction of $\A$, the probability that it succeeds in
  cloning $q+1$ money states from $q$ copies is at least
  $\secpar^{-c}$.
\end{proof}

Our security proof of the quantum money scheme is arguably simpler
than that in~\cite{AC12}.
In~\cite{AC12}, to prove their hidden subspace money scheme is secure,
one needs to develop the so called inner-product adversary method to
show the worst-case query complexity for the hidden subspace states
and use a random self-reducible argument to establish the average-case
query complexity.
In our case, it follows almost directly from the cryptographic
non-cloning theorem with oracles.
The quantum money schemes derived from PRS's enjoy many nice features
of the hidden subspace scheme.
Most importantly, they are also \emph{query-secure}~\cite{AC12},
meaning that the bank can simply return the money state back to the
user after verification.

It is also interesting to point out that quantum money states are not
necessarily pseudorandom states.
The hidden subspace state~\cite{AC12}, for example, do not satisfy our
definition of PRS as one can measure polynomially many copies of the
state in the computational basis and recover a basis for the hidden
subspace with high probability.

\section{Pseudorandom Unitary Operators}
\label{sec:psu}


Let $\H$ be a Hilbert space and let $\K$ a key space, both of which
depend on a security parameter $\secpar$.
Let $\mu$ be the Haar measure on the unitary group $\Unitary(\H)$.

\begin{definition}
  \label{def:pru}
  A family of unitary operators $\{U_k \in\Unitary(\H)\}_{k\in\K}$ is
  \textbf{pseudorandom}, if two conditions hold:
  \begin{enumerate}
  \item (\textbf{Efficient computation}) There is an efficient quantum
    algorithm $Q$, such that for all $k$ and any
    $\ket{\psi}\in\State(\H)$, $Q(k, \ket{\psi}) = U_k \ket{\psi} $.
  \item (\textbf{Pseudorandomness}) $U_k$ with a random key $k$ is
    \textbf{computationally indistinguishable} from a Haar random
    unitary operator.
    More precisely, for any efficient quantum algorithm $\A$ that
    makes at most polynomially many queries to the oracle,
    \begin{equation*}
      \abs{\Pr_{k\gets \K} \bigl[ \A^{U_k}(1^\secpar) =1 \bigr] -
        \Pr_{U\gets \mu} \bigl[\A^U(1^\secpar) = 1 \bigr]} =
      \negl(\secpar).
    \end{equation*}
  \end{enumerate}
\end{definition}


Our techniques for constructing pseudorandom states can be extended to
give candidate constructions for pseudorandom unitary operators (PRUs)
in the following way.
Let $\H = (\mathbb{C}^2)^{\otimes n}$.
Assume we have a pseudorandom function $\prf:\: \K \times \X \rightarrow \X$, with
domain $\X = \{0, 1, 2, \ldots, N-1\}$ and $N=2^n$.
Using the phase kick-back technique, we can implement the unitary
transformation $T_k \in \Unitary(\H)$ that maps
\begin{equation}
  T_k:\: \ket{x} \mapsto \omega_N^{\prf_k(x)} \ket{x}, \quad \omega_N = \exp(2\pi i/N).
\end{equation}
Our pseudorandom states were given by $\ket{\phi_k} = T_k H^{\otimes n}
\ket{0}$, where $H^{\otimes n}$ denotes the $n$-qubit Hadamard transform.
We conjecture that by repeating the operation $T_k H^{\otimes n}$ a
polynomial number of times, we get a pseudorandom unitary operation.
Alternatively, a polynomial number of times of repetition of the
circuit
\begin{equation*}
  \prp_k (H^{\otimes n}\otimes \I^{\otimes n}),
\end{equation*}
on $2n$ qubits may be another candidate construction.


Pseudorandom unitary operators may be useful and provide efficient
alternatives for unitary $t$-designs in may settings.
It also trivially implies the existence of PRS's by definition and can
be employed whenever a PRS is used.

\begin{lemma}
  For any PRU $\{ U_k \}_{k\in\K}$, $\{ U_k \ket{0} \}_{k\in\K}$ is a
  PRS.
\end{lemma}
\begin{proof}
  If there is a quantum algorithm $\A$ that takes $m\in\poly(\secpar)$
  copies of state $U_k \ket{0}$ and distinguishes it from Haar random
  states, we can design algorithm $\A'$ that distinguishes $ \{ U_k
  \}$ from the Haar random $U$ in the following way.
  $\A'$ prepares $m$ copies of $\ket{0}$, calls the unitary oracle on
  each copy, sends these states to algorithm $\A$ and outputs whatever
  $\A$ outputs.
\end{proof}

\newpage

\bibliographystyle{acm}

\bibliography{pseudorandom-quantum-states}


\end{document}